%% file: main_CDC_v2_arXiv.tex
\documentclass[letterpaper, 10 pt, conference]{ieeeconf}  

\IEEEoverridecommandlockouts                              

\usepackage{amsmath} 
\usepackage{amssymb}  

\usepackage{amssymb}
\usepackage{amsmath}
\usepackage{color}
\usepackage{graphicx}
\usepackage{dsfont}
\usepackage{mathrsfs}
\usepackage{etoolbox}
\usepackage{float}
\usepackage{subfigure}
\usepackage{amsbsy}
\usepackage{mathabx}
\usepackage{comment}
\usepackage{hyperref}
\usepackage{algorithm}
\usepackage{algpseudocode}
\usepackage{caption}
\usepackage{relsize}
\usepackage{siunitx}
\usepackage{soul}
\usepackage[noadjust]{cite}
\usepackage{textcomp}
\usepackage{multirow}
\usepackage[normalem]{ulem}

\definecolor{aoenglish}{rgb}{0.0, 0.5, 0.0}
\definecolor{darkblue}{rgb}{0.0, 0.0, 0.55}
\definecolor{darkmagenta}{rgb}{0.55, 0.0, 0.55}
\definecolor{electricviolet}{rgb}{0.56, 0.0, 1.0}
\definecolor{electricyellow}{rgb}{1.0, 1.0, 0.0}
\definecolor{forestgreen}{rgb}{0.13, 0.55, 0.13}
\definecolor{fuchsia}{rgb}{1.0, 0.0, 1.0}
\definecolor{gamboge}{rgb}{0.89, 0.61, 0.06}
\definecolor{goldenpoppy}{rgb}{0.99, 0.76, 0.0}
\definecolor{indigo}{rgb}{0.29, 0.0, 0.51}
\definecolor{internationalorange}{rgb}{1.0, 0.31, 0.0}
\definecolor{lava}{rgb}{0.81, 0.06, 0.13}
\definecolor{selectiveyellow}{rgb}{1.0, 0.73, 0.0}
\definecolor{turquoiseblue}{rgb}{0.0, 1.0, 0.94}
\definecolor{turquoise}{rgb}{0.19, 0.84, 0.78}

\definecolor{modcol}{RGB}{255,102,178}

\newcommand{\rev}[1]{\textcolor{black}{#1}}

\definecolor{MFabg}{RGB}{76, 191, 229}

\newtheorem{thm}{Theorem}

\input{shortcuts_math_CDC}

\title{\LARGE \bf On the impact of regularization in data-driven predictive control}

\author{Valentina~Breschi~\IEEEmembership{Member,~IEEE,} Alessandro~Chiuso~\IEEEmembership{Fellow,~IEEE,}\\ Marco~Fabris~\IEEEmembership{Member,~IEEE,} and~Simone~Formentin~\IEEEmembership{Member,~IEEE}
\thanks{This project was partially supported by the Italian Ministry of University and Research under the PRIN'17 project \textquotedblleft Data-driven learning of constrained control systems \textquotedblright, contract no. 2017J89ARP.}
	\thanks{\rev{A. Chiuso and M. Fabris are with the Department of Information Engineering, University of Padova, Padua, Italy.
	V. Breschi is with Department of Electrical Engineering, Eindhoven University of Technology, Eindhoven, Netherlands.
	S. Formentin is with Dipartimento di Elettronica, Informatica e Bioingegneria (DEIB), Politecnico di Milano, Milan, Italy.} 
}
	\thanks{\mbox{Corresponding author: S. Formentin, {\tt \scriptsize \href{mailto:simone.formentin@polimi.it}{simone.formentin@polimi.it}}}
	}
}

\begin{document}
	
\maketitle
\thispagestyle{empty}
\pagestyle{empty}

\begin{abstract}
Model predictive control (MPC) is a control strategy widely used in industrial applications. However, its implementation typically requires a mathematical model of the system being controlled, which can be a time-consuming and expensive task. Data-driven predictive control (DDPC) methods offer an alternative approach that does not require an explicit mathematical model, but instead optimize the control policy directly from data. In this paper, we study the impact of two different regularization penalties on the closed-loop performance of a recently introduced data-driven method called $\gamma$-DDPC. Moreover, we discuss the tuning of the related coefficients in different data and noise scenarios, to provide some guidelines for the end user.
	
\end{abstract}



\section{Introduction} 
\label{sec:intro}
Model predictive control (MPC) is a popular control strategy that has been successfully applied in a wide range of applications \cite{borrelli2017predictive}. However, a major limitation of MPC is that it requires a mathematical model of the system being controlled, which can be a costly and time-consuming task. This requirement has led to the development of data-driven predictive control (DDPC) methods, which aim to learn the control policy directly from data without the need for a mathematical model of the plant \rev{\cite{Baros2022onlineDeePC}} \cite{berberich2020data,breschi2022design,sassella2022data}.

Nonetheless, the data-based predictor used in DDPC is not exempt from shortcomings, due to the 
presence of noise on the measured data. Therefore, different techniques have been proposed to make the closed-loop performance less sensitive to such a noise, e.g., robust design in case hard power bounds are given \cite{berberich2020data}, dynamic mode decomposition \cite{sassella2022noise} and regularization \cite{dorfler2022bridging}. The latter in particular can be used to prevent the data-based predictor to overfit the historical data, by tuning a few penalty coefficients. In the pioneering work \cite{dorfler2022bridging}, the design of such terms is discussed for different kinds of regularization, and the authors highlight 
the significant efforts required in terms of trial-and-error tuning, especially as far as some specific parameters are concerned. In \cite{breschi2022roleArXiv}, we showed that regularization may be avoided in case the data set is large enough and the DDPC problem is reformulated thanks to subspace identification tools, so as to shrink the number of decision variables, into the so-called $\gamma$-DDPC method. Finally, in \cite{breschi2022uncertainty}, we have 
focused on 
finite size 
data sets and used 
asymptotic arguments to show that regularization might instead be useful to counteract the prediction error variance, due to the use of noisy data in the predictor. Two different regularization options have been introduced, and an on-line tuning of 
the associated penalizations has been proposed, based on the prior knowledge of 
the variance expression.

This paper\rev{'s contribution is built upon} 
\cite{breschi2022uncertainty}, since our goal here is to analyze the joint tuning of the two regularization terms of $\gamma$-DDPC and analyze their impact on the closed-loop performance. In particular, we shall discuss the role of the driving input color (spectra) and some qualitative guidelines about regularization design will be drawn by means of extensive simulations on a benchmark linear system as well as on a challenging nonlinear problem, namely, wheel slip control in braking maneuvering. Finally, offline and on-line regularization tuning will be compared.

The remainder of the paper is as follows. In Section \ref{sec:setting}, the predictive control problem setting is described, and the regularization tuning issue is mathematically formulated. Section \ref{sec:reg} illustrates the considered regularization techniques for $\gamma$-DDPC and discusses the role of each term, also by means of two numerical case studies. The paper is ended by some concluding remarks.

\textit{Notation.} Given a signal (say $u(t) \in \mathbb{R}^m$\rev{)}, the associated (block) Hankel matrix $U_{[t_0,t_1],N} \in \mathbb{R}^{m(t_1-t_0+1) \times N}$ is defined as:
\begin{equation}\label{eq:Hankel}
	U_{[t,s],N}\!:=\!\!\frac{1}{\sqrt{N}}\!\begin{bmatrix}
		u(t) & u(t\!+\!1) & \cdots & u(t\!+\!N\!-\!1)\\
		u(t\!+\!1) & u(t\!+\!2) & \cdots & u(0\!+\!N)\\
		\vdots & \vdots & \ddots & \vdots\\
		u(s) &u(s\!+\!1) & \dots & u(s\!+\!N\!-\!1)  
	\end{bmatrix}\!\!,
\end{equation}
while we use the shorthand $U_{t}:= U_{[t,t],N}$ to denote a single (block) row Hankel, namely:
\begin{equation}\label{eq:Hankel:onerow}
	U_{t}:= \frac{1}{\sqrt{N}}\begin{bmatrix}
		u(t) & u(t\!+\!1) & \cdots & u(t\!+\!N\!-\!1)
	\end{bmatrix}. 
\end{equation}

\section{Problem setting}\label{sec:setting}
Our goal is to design a controller for an 
\emph{unknown} plant that can be modeled by \emph{linear time-invariant} (LTI) discrete-time linear (stochastic) system $\mathcal{S}$.  Without loss of generality, we consider its state space description  in 
%
%
\emph{innovation form}
\begin{equation}\label{eq:stoc_sys}
	\begin{cases}
		x(t+1)=Ax(t)+Bu(t)+Ke(t)\\
		y(t)=Cx(t)+Du(t)+e(t), 	\end{cases} \quad  t \in \mathbb{Z},
\end{equation}
where $x(t)\in \mathbb{R}^{n}$, $u(t) \in \mathbb{R}^{m}$ and $e(t) \in \mathbb{R}^{p}$ are the state, input and  innovation process  respectively, while $y(t) \in \mathbb{R}^{p}$ is the corresponding output signal.

Under the \emph{unrealistic} assumption that the system matrices $(A,B,C,D,K)$ are known, the predictive constrained tracking control problem of interest for this paper (for a given reference $y_{r}(t)$ and a prediction horizon $T$) can be formulated as follows 
\begin{subequations}\label{eq:RHPC_prob}
	\begin{align}
		&\underset{\{u(k)\}_{t}^{t+T-1}}{\mbox{minimize}}~\frac{1}{2}\sum_{k=t}^{t+T-1} \ell(u(k),\hat{y}(k),y_{r}(k)) \label{eq:cost}\\
		& \mbox{s.t. } \hat x(k\!+\!1)\!=\!\!A\hat x(k)\!+\!Bu(k),~k \!\in\! [t,t\!+\!T),\\
		& \qquad  \hat y(k)\!=\!C\hat x(k)+Du(k),~k \in [t,t+T),\\ 		
		&\qquad \hat x(t) = \hat x_{init},\\ 
		&\qquad u(k) \in \mathcal{U},~\hat y(k) \in \mathcal{Y},~k \in [t,t+T),
	\end{align}
	where {{$k \in \mathbb{Z}$}}, $\hat x_{init}$ is the state-estimate at time $t$, which can be obtained by running a conventional Kalman filter given the input-output measurements available up to time $t$, and  $y_r$ is the reference signal. 
	\rev{Also}, $\ell(\cdot)$ is a convex loss function, penalizing both the tracking performance and the control effort, \emph{e.g.,} 
	\begin{equation}\label{eq:loss}
		\ell(u(k),\hat{y}(k),y_{r}(k))=\|\hat{y}(k)\!-\! y_{r}(k)
		\|_{Q}^{2}\!+\!\|u(k)\rev{-u_r(k)}\|_{R}^{2},
	\end{equation}
\end{subequations}
where the penalties $Q \in \mathbb{R}^{p \times p}$ and $R \in \mathbb{R}^{m \times m}$, with $Q \succeq 0$ and $R \succ 0$, are selected to trade-off between tracking performance and control effort. 

A standard assumption in \emph{data-driven} control is that the system matrices $(A,B,C,D,K)$ are \emph{not known}, and only a  finite sequence of input/output data $\mathcal{D}_{N_{data}}=\{u(j),y(j)\}_{j=1}^{N_{data}}$. We would like to stress that in our  framework measured data are by assumption noisy, in the sense that there is no LTI system that, with the given input $u(t)$, produces exactly the measured output \rev{in a deterministic way}. 

In this paper, we follow that data-driven predictive control problem formulation provided in \cite{breschi2022role,breschi2022uncertainty}, and we refer to those papers for a connection with the recent related literature such as \cite{dorfler2022bridging,berberich2020data}.

To this purpose, we need to introduce the Hankel matrices, including past and future values of inputs and outputs, with respect to time $t$. In particular, with obvious use of the subscripts $P$ and $F$, we define:
\begin{align}\label{eq:future}
	U_F\!:=&U_{[\rho,\rho+T-1],N},~Y_F\!:=\! Y_{[\rho,\rho+T\!-\!1],N},
\end{align}
where $N:=N_{data}-T-\rho$ and $\rho$ is the \textquotedblleft past horizon\textquotedblright. 

Based on \eqref{eq:stoc_sys} the Hankel $Y_F$ can be written as
\begin{subequations}
	\begin{equation}
		{Y}_{F}=\Gamma X_{\rho} +\mathcal{H}_{d}U_{F}+ \mathcal{H}_{s}E_F, 
	\end{equation}
	where $E_F$ is the Hankel of future innovations,
	\begin{equation}\label{eq:observability}
		\Gamma=\begin{bmatrix} C \\ CA \\ CA^{2}\\ \vdots \\ CA^{T-1} \end{bmatrix},
	\end{equation}   
	and $\mathcal{H}_{d} \in \mathbb{R}^{pT \times mT}$ and $\mathcal{H}_s \in \mathbb{R}^{pT \times pT}$ are the Toeplitz matrices formed with the Markov parameters of the system, namely
	\begin{align}
		& 	\mathcal{H}_{d} = \begin{bmatrix} 
			D & 0  & 0 & \dots & 0  \\
			CB & D  &  0 &\dots & 0 \\
			CAB & CB & D  &  \dots & 0 \\
			\vdots & \vdots  & \vdots &  \ddots & \vdots & \\
			CA^{T-2}B & CA^{T-3}B  & CA^{T-4}B & \ldots &D 
		\end{bmatrix},\\
		&\mathcal{H}_s = \begin{bmatrix} 
			I & 0  & 0 & \dots & 0 \\
			CK & I  &  0 &\dots & 0 \\
			CAK & CK & I  &  \dots & 0 \\
			\vdots & \vdots  & \vdots &  \ddots & \vdots  \\
			CA^{T-2}K & CA^{T-3}K  & CA^{T-4}K & \ldots &I 
		\end{bmatrix}. 
	\end{align}
\end{subequations}

Let us now define $z(t)$ as the joint input/output process
\begin{equation}\label{eq:z}
	z(t):=\begin{bmatrix} u(t)\\
		y(t)
	\end{bmatrix},
\end{equation}
with the associated Hankel matrix being $Z_P\!\!:=\!Z_{[0,\rho-1],N}$. The orthogonal projection of $Y_F$ onto the row space of $Z_P$ and $U_F$ turns out to be given by
\begin{align}
	\hat{Y}_{F}&=\Gamma \hat X_{\rho} +\mathcal{H}_{d}U_{F}+ \underbrace{\mathcal{H}_{s}\Pi_{Z_{P},U_{F}}(E_F)}_{O_P(1/\sqrt{N})} \label{eq:Proj:F}
\end{align}
where the last term vanishes\footnote{For a more formal statement on this, we refer the reader to standard literature on subspace identification.} (in probability) as $1/\sqrt{N}$.
This means that, when the matrices $(A,B,C,D,K)$ are \emph{unknown}, future outputs can still be predicted directly from data. In fact, given any (past) joint input and output trajectory and future control inputs 
\begin{equation}\label{eq:zinit}
	\begin{matrix}
		z_{init}:=\begin{bmatrix}
			z(t-\rho)\\
			\vdots\\
			z(t-2)\\
			z(t-1)
		\end{bmatrix}, \quad & u_{f}:=\begin{bmatrix}
			u(t)\\
			u(t+1)\\
			\vdots\\
			u(t+T-1)
		\end{bmatrix},
	\end{matrix}
\end{equation}
the prediction $\hat y_f$ of future outputs $y_f$ 
\begin{equation}
	y_{f}:=\begin{bmatrix}
		y(t)\\
		y(t+1)\\
		\vdots\\
		y(t+T-1)
	\end{bmatrix},
\end{equation}
based on past inputs $z_{init}$ and future inputs $u_f$ 
can be obtained from\footnote{Conditions on $
	\rho$ for this to hold are provided in \cite{breschi2022role}.}
\begin{equation}\label{eq:DDPC_standard}
	\begin{bmatrix}
		z_{init}\\u_{f}\\ \hat y_{f}			
	\end{bmatrix}=\begin{bmatrix}
		Z_{P}\\U_{F}\\ \hat Y_{F}
	\end{bmatrix}\alpha + O_P(1/
	\sqrt{N}), 
\end{equation}
with $\alpha \in \mathbb{R}^{N}$ to be optimized as in, e.g., \cite{berberich2020data}, \cite{breschi2022role}, \cite{coulson2019data}. 

Following subspace identification \cite{Vanov-book} ideas,  the orthogonal projection \eqref{eq:Proj:F} can be written  exploiting the LQ decomposition of the data matrices. In particular, let us define 
\begin{equation}\label{eq:LQ}
	\begin{bmatrix}
		Z_{P}\\U_{F}\\ Y_{F}
	\end{bmatrix}=
	\begin{bmatrix}
		L_{11} & 0 & 0  \\
		L_{21} & L_{22} &  0\\
		L_{31} & L_{32} & L_{33} 
	\end{bmatrix}
	\begin{bmatrix}
		Q_{1}\\
		Q_{2}\\
		Q_{3}
	\end{bmatrix}.
\end{equation}
where the matrices $\{L_{ii}\}_{i=1}^{3}$ are all non-singular and $Q_{i}$ have orthonormal rows, i.e., $Q_{i}Q_{i}^{\top}=I$, for $i=1,\ldots,3$, $Q_i Q_j^\top = 0$, $i\neq j$. The orthogonal projection 
\eqref{eq:Proj:F}  can be written in the form:
$$
\hat Y_F = L_{31} Q_1 + L_{32}Q_2
$$
With this notation, following the same rationale of \cite{breschi2022role,breschi2022uncertainty}, we can further reformulate \eqref{eq:DDPC_standard} as:
\begin{equation}\label{eq:LQ_prediction}
	\begin{bmatrix}z_{init}\\u_{f}\\ \hat y_{f}		
	\end{bmatrix}=\begin{bmatrix}
		Z_{P}\\U_{F}\\ \hat Y_{F}
	\end{bmatrix}\alpha=
	\begin{bmatrix}
		L_{11} & 0   \\
		L_{21} & L_{22} \\
		L_{31} & L_{32} 
	\end{bmatrix}
	\underbrace{\begin{bmatrix}
			Q_{1}\\
			Q_{2}
		\end{bmatrix}\alpha}_{\gamma} + O_P(1/
	\sqrt{N}).
\end{equation}

and the parameters
\begin{equation}\label{eq:gamma}
	\gamma=\begin{bmatrix}
		\gamma_{1}\\
		\gamma_{2}
	\end{bmatrix},
\end{equation}
become the new decision variables. In addition, in \cite{breschi2022uncertainty} it was suggested to add a (slack) optimization variable $\gamma_3$ to model the projection error in \eqref{eq:Proj:F} and avoid overfitting. In particular, the prediction  (with slack) can be written as:
$$
\bar y_f = \underbrace{\begin{bmatrix}
		L_{31} & L_{32} 
	\end{bmatrix} \begin{bmatrix}
		\gamma_{1}\\\gamma_{2}
\end{bmatrix}}_{=\hat y_f} + L_{33}\gamma_3
$$

We refer the reader to \cite{breschi2022uncertainty} for a sound statistical motivation of this particular expression of the slack $L_{33}\gamma_3$. In particular, since $L_{33}$ is generically of full rank, constraints/regularization should be imposed on the slack optimization variable $\gamma_3$.


A data-driven predictive controller with the same objectives and constraints of \eqref{eq:RHPC_prob} can be formulated as follows \cite{breschi2022role}
\begin{subequations}\label{eq:RHPC_prob_dd_gamma}
	\begin{align}
		&\underset{\gamma_2,\gamma_3}{\mbox{min}}~\frac{1}{2}\sum_{k=t}^{t+T-1} \ell(u(k),\bar y(k),y_{r}(k)) + \Psi(\gamma_1,\gamma_2,\gamma_3) \label{eq:cost_gammaDDPC}\\
		&~~\mbox{s.t.}~~\begin{bmatrix}
			u_{f}\\
			\bar y_{f}
		\end{bmatrix}=\begin{bmatrix}
			L_{21} & L_{22} & 0 \\
			L_{31} & L_{32} & L_{33}
		\end{bmatrix}\begin{bmatrix}
			\gamma_{1}^\star\\\gamma_{2} \\ \gamma_3
		\end{bmatrix} \label{eq:prediction_model3},\\
		&~~~~~~~~~u(k) \in \mathcal{U},~\bar y(k) \in \mathcal{Y},~k \in [t,t+T), \label{eq:constraints2}
	\end{align}
\end{subequations}
with 
\begin{equation}\label{eq:loss}
	\ell(u(k),\bar y(k),y_{r}(k))=\|\bar y(k)\!-\! y_{r}(k)
	\|_{Q}^{2}\!+\!\|u(k)\|_{R}^{2},
\end{equation}
and
\begin{equation}\label{eq:init_terms}
	\gamma_{1}^\star=
	L_{11}^{-1}z_{init},
\end{equation}
where $z_{init}$ is defined as in \eqref{eq:zinit} and the choice of $\gamma_{1}$ straightforwardly follows from the initial conditions (showing the advantages of using $\gamma$ instead of $\alpha$ as the decision vector). 



The purpose of this paper is \textit{to study the design and impact  of the regularization term $\Psi(\gamma_1,\gamma_2,\gamma_3)$ within a noisy stochastic environment, and provide the end user with useful hints on how to tune such a penalty term}.

\section{The role of regularization}\label{sec:reg}

In \cite{breschi2022uncertainty}, it has been argued that the average variance of the error on the future output predictions $\hat y_f$ due to the finite data projection errors in \eqref{eq:Proj:F}, is proportional to $\|\gamma_1\|^2  + \|\gamma_2\|^2$. Since, in the optimization problem \eqref{eq:RHPC_prob_dd_gamma}, $\gamma_1$ is determined by the initial conditions, it only remains to regularize $\gamma_2$ so as to avoid an (unnecessarily) high variance on the predictor and, therefore, poor control performance. In this paper, we consider also an alternative regularization term that penalizes directly the control input effort (in addition to the control penalty already embedded in the control cost), and discuss its relation with regularization on $\gamma_2$. Differently from  \cite{breschi2022uncertainty}, we consider this jointly with presence of a slack variable $\gamma_3$ and thus a related regularization. These considerations lead to the following two forms  of the regularization term $\Psi(\gamma_1,\gamma_2,\gamma_3)$ in 
\eqref{eq:RHPC_prob_dd_gamma}:

\begin{enumerate}
	\item[(a)] {\bf Regularization on $\gamma_2$ and slack $\gamma_3$} 
	\begin{equation}\label{eq:reg1}
		\Psi_{\gamma_2}(\gamma_1,\gamma_2,\gamma_3) := \beta_2\|\gamma_2\|^2 + \beta_3 \|\gamma_3\|^2 ;
	\end{equation}
	\item[(b)] {\bf Regularization on input $u_f$ and slack $\gamma_3$} 
	\begin{align}\label{eq:reg2}
		\Psi_u(\gamma_1,\gamma_2,\gamma_3) :=& \beta_2\|u_f\|^2 + \beta_3 \|\gamma_3\|^2 \nonumber \\
		 = & \beta_2\|L_{21} \gamma_1 + L_{22} \gamma_2\|^2 + \beta_3 \|\gamma_3\|^2 ;
	\end{align}
	
\end{enumerate}

\rev{where $(\beta_{2},\beta_{3})$ are hyper-parameters to be determined.}

\subsection{Theoretical analysis}
We first state a Theorem the establishes the connection between \eqref{eq:reg1} and \eqref{eq:reg2}.

\begin{thm}\label{thm:reg}
	If the training input sequence $u(t)$ in the Hankel matrices $U_F$ and $U_P$ is (zero mean) white with variance $\sigma^2 I$, the regularization terms $\Psi_{\gamma_2}$ in \eqref{eq:reg1} and $\Psi_{u}$ in \eqref{eq:reg2} are asymptotically (in $N$) equivalent up to a rescaling of the weight $\beta_2$.
\end{thm}

\begin{proof}
	Under the assumption that $u(t)$ is white noise, then the future inputs are uncorrelated with past input and output data, so that the projection $\hat U_F:=\Pi_{Z_{P}}(U_F)$ of $U_F$ on the joint past $Z_P$ tends to zero as \rev{$O_P(1/\sqrt{N})$}, more precisely
 
	\begin{equation}\label{eq:L21}
		\hat U_F: = L_{21}Q_1.
	\end{equation}
	Since $Q_1Q_1^\top = I$, it follows that $L_{21} = O_P(1/\sqrt{N})$. In addition, since $u$ is white, its sample covariance matrix $U_FU_F^\top$ converges to $\sigma^2 I$, i.e. 
	\begin{equation}\label{eq:L22}
		U_FU_F^\top = \underbrace{L_{21}L_{21}^\top}_{O_P(1/{N})} + L_{22}L_{22}^\top \mathop{\longrightarrow}^{N \rightarrow\infty} \sigma^2  I
	\end{equation}
	Equations \eqref{eq:L21} and \eqref{eq:L22} imply that, asymptotically in $N$, $L_{21}\simeq 0$ and $L_{22}\simeq \sigma I$. Therefore we have:
	\begin{equation}\label{eq:reg2eqreg1}
		\begin{array}{rcl}
			\Psi_u(\gamma_1,\gamma_2,\gamma_3) &:=&  \beta_2\|L_{21} \gamma_1 + L_{22} \gamma_2\|^2 + \beta_3 \|\gamma_3\|^2\\
			& \simeq & \beta_2 \sigma^2\|\gamma_2\|^2 + \beta_3 \|\gamma_3\|^2
		\end{array}
	\end{equation}
	showing that, up to the rescaling of the weight $\beta_2$, this is equivalent to $\Psi_{\gamma_2}(\gamma_1,\gamma_2,\gamma_3) $
\end{proof}

This result has two important implications:
\begin{itemize}
	\item  When the (training) input is white, regularization on $\gamma_2$ is equivalent to a penalty on the future input energy, which is typically present in the control cost. As such, we can  argue that, in this case, the control cost has an \textit{indirect but important} effect in counteracting the effect of the noise variance in the predictor.
	\item When the training input is not white, the control energy cost \emph{is not} equivalent to penalizing the norm of $\gamma_2$, which on the other hand should be penalized to limit the effect of noise variance. The simulation results in the next section indeed confirm that, when noise input is not white, regularization on $\gamma_2$ (i.e. $\Psi_{\gamma_2}$) has to be \rev{included}.
\end{itemize}
\subsection{Experimental analysis}\label{sec:examples}
In this section we shall illustrate, exploiting two numerical examples (one linear and one nonlinear\footnote{\rev{By working in a specific operating regime, the control of an unknown nonlinear system can be tackled as that of an uncertain linear system.}}), the role of different  regularization terms in the optimal control problem \eqref{eq:RHPC_prob_dd_gamma}. In particular,  following the rationale proposed in \cite{dorfler2022bridging}, we evaluate the  closed-loop performance over $T_v$ feedback steps as measured by the performance index:
\begin{equation}\label{eq:general_perf_index}
	\!\! J(u,y)\!=\!  \dfrac{1}{T_{v}}\!\sum\limits_{t=0}^{T_{v}-1}\!\!\left( \left\|u(t)\!-\!u_r(t) \right\|^{2}_{R} \!+\! \left\|y(t)\!-\!y_{r}(t) \right\|^{2}_{Q}\right) .
\end{equation}
\begin{figure}[h!]
	\centering
	\subfigure[Performance indexes]{\includegraphics[height=0.247\textwidth, trim={1.5cm 0cm 3cm 1cm},clip]{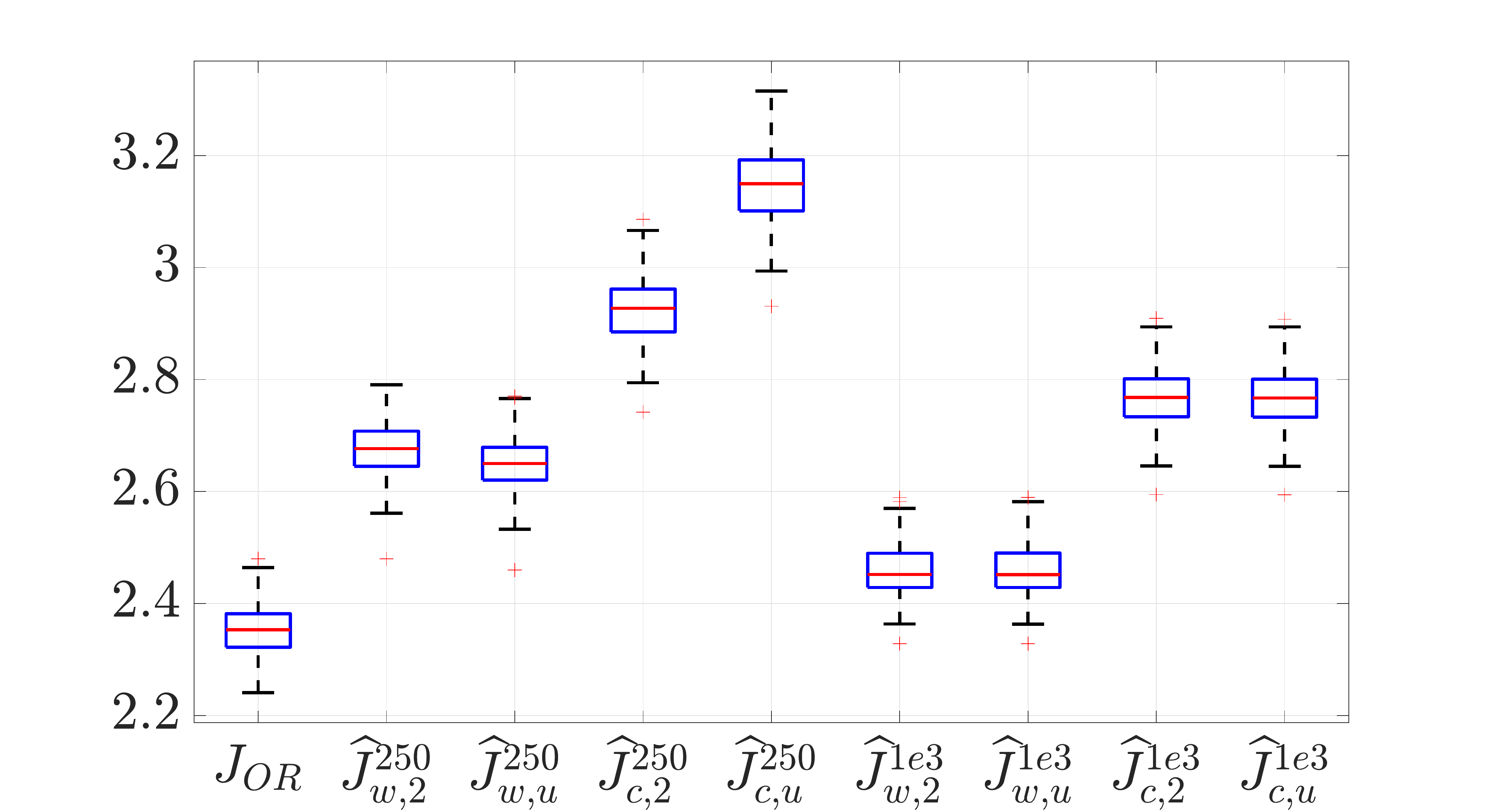}}
	\subfigure[Optimal performance under constraint $\beta_2=0$]
	{\includegraphics[height=0.26\textwidth,trim={0cm 0cm 2cm 0cm},clip]{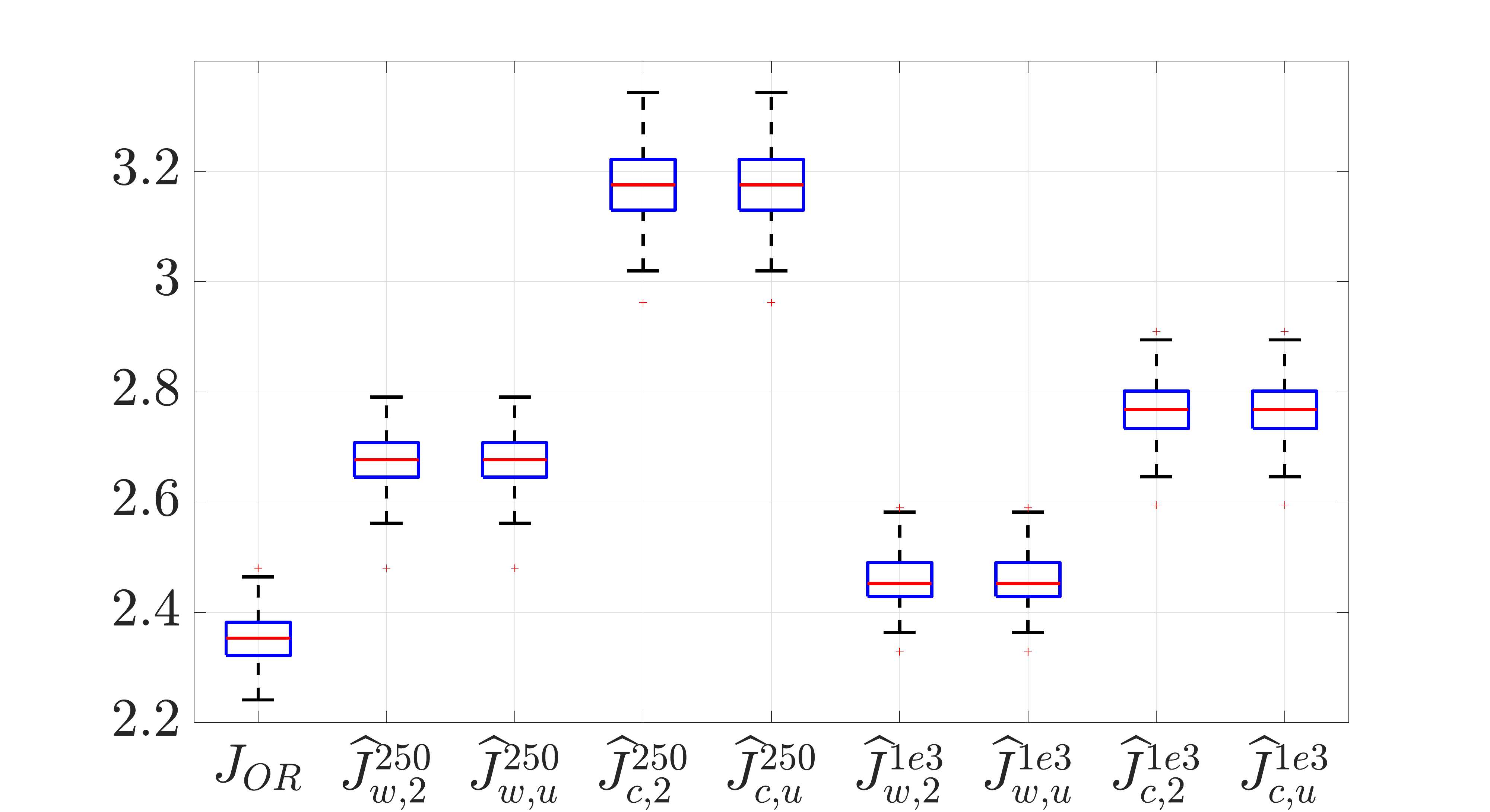}\label{fig:performances_b2=0}}
	\subfigure[Corresponding minimizers]{\includegraphics[height=0.28\textwidth,trim={0cm 0cm 2cm 1cm},clip]{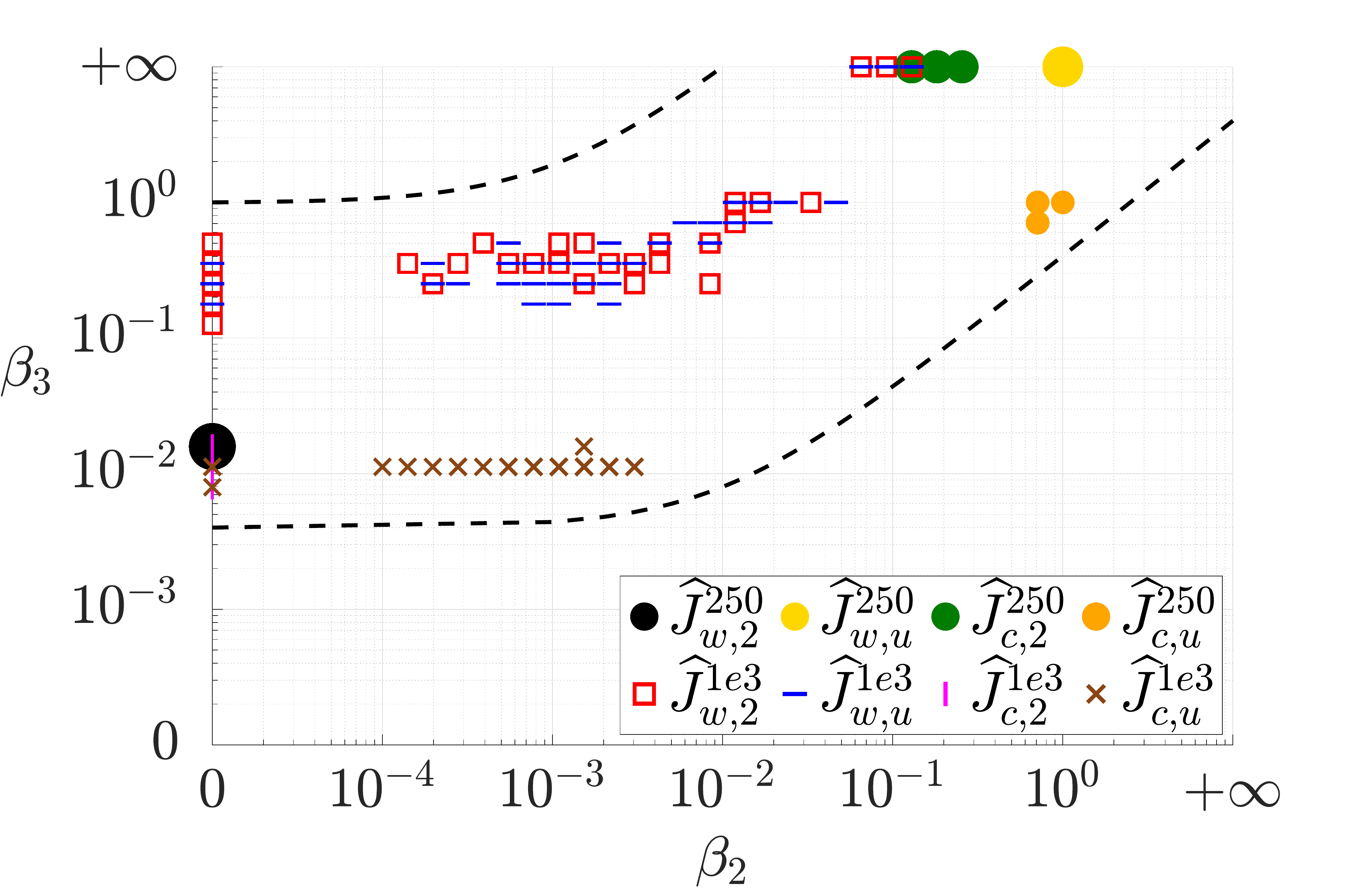}\label{fig:Markovski_scatter}}
	\caption{(a): comparison between the Kalman-filter-based oracle performance $J_{OR}$ and the minimum cost realizations $\widehat{J}^{N_{data}}_{n_s,rg}$ for $\Sigma_{L}$ over $100$ Monte Carlo runs;  
	\rev{(b):} Optimal performance under the constraint $\beta_2=0$; (c): distribution of the corresponding minimizers $(\beta_2^\star,\beta_3^\star)$.}
	\label{fig:Markovski_costs}
\end{figure}

\subsubsection{Benchmark LTI system}
Consider the SISO, 
$4$-th order 
system in \cite{LandauReyKarimi1995} ($\Sigma_{L}$ in the sequel) with a prediction horizon $T=20$. 
To assess the impact of the \emph{training} data  on closed-loop performance, we consider four data sets of two different lengths $N_{data}$ (either $250$ or $1000$), obtained either with white noise input (denoted with $n_s=w$) or with a low-pass   (obtained filtering white noise with  a discrete-time low-pass filter with cut-off angular frequency $1.8~rad/s$) input sequence (denoted with $n_s=c$).
White noise is added to the output to guarantee  a signal-to-noise ratio of $15$~dB.

The data-driven optimal control problem \eqref{eq:RHPC_prob_dd_gamma} is solved  
fixing $Q = 10^{3}$ and $R = 10^{-2}$,
setting the output reference $y_{r}(t) = \sin(5\pi t/(T+T_{v}-1))$ and the input reference $u_r(t) = 0$, with $T_{v} = 50$. The two different regularization strategies discussed in Section \ref{sec:reg} are denoted with the shorthands $rg=2$ 
\rev{and $rg=u$ for \eqref{eq:reg1} and \eqref{eq:reg2}, respectively.}

The ``oracle" value of $(\beta_2,\beta_3)$ leading to \textbf{the minimum cost \eqref{eq:general_perf_index}}, here denoted as $J_{n_s,rg}^{N_{data}}= J(u^{N_{data}}_{n_s,rg},y^{N_{data}}_{n_s,rg})$ to account for the different data set lengths, input and regularization strategies, are searched over a rectangular logarithmic-spaced grid $G_{23}$ with $7$ points per decade, 
so that $G_{23} \subseteq \left\lbrace 0 \right\rbrace \cup [10^{-4}, 10^{0}] \cup \left\lbrace +\infty \right\rbrace \times \left\lbrace 0 \right\rbrace \cup [10^{-3},10^{0}] \cup \left\lbrace +\infty \right\rbrace$.

We perform $100$ Monte Carlo experiments\footnote{\rev{The past horizon $\rho$ is determined using Akaike's information criterion.}} (i.e. $100$ different training data sets with the output of  $\Sigma_L$ corrupted by white noise) to tune $\beta_2$ and $\beta_3$
for all the four considered training scenarios and possible regularization. For each Monte Carlo run, and for each set of possible parameters in the grid $G_{23}$,  the closed loop performance index  \eqref{eq:general_perf_index} is computed by averaging over $100$ closed loop experiments  (all with the same control law but different closed measurements errors) the corresponding performance index $J^{N_{data}}_{n_s,rg}(i)$, i.e.  
\begin{equation}\label{eq:avg_index}
	{\widebar{J}^{N_{data}}_{n_s,rg}} = \frac{1}{100}\sum\nolimits_{i=1}^{100}	J^{N_{data}}_{n_s,rg}(i),
\end{equation}
The  optimal values of $\beta_2$ and $\beta_3$  over the grid $G_{23}$ is obtained by finding the minimum $\widehat{J}^{N_{data}}_{n_s,rg} {=\min_{(\beta_2,\beta_3) \in G_{23}} \{\widebar{J}^{N_{data}}_{n_s,rg}\}}$.

The results are reported  in Fig. \ref{fig:Markovski_costs}, based on which  we can make the following general considerations:
\begin{itemize}
	\item As expected based on Theorem \ref{thm:reg}, when the input is white noise, the two types of regularization provide the same performance (minor differences for $N_{data}=250$ are due to  sample variability).
	\item The ``optimal'' (oracle) closed loop performance obtained with the two different regularization strategies differ when the input is not white. In particular, the penalty \eqref{eq:reg1} that acts directly on $\gamma_2$ and thus controls the predictor variance provides the best performance, particularly so for small data sets where the effect of noise has more impact.  
	\item Based on the comparison between Fig.\ref{fig:Markovski_costs}(a) and Fig. \ref{fig:Markovski_costs}\rev{(b)} 
	(in which $\beta_2$ is constrained to zero), we can observe that the impact of $\beta_2$ is significant in the low-data regime (equivalent to large noise in the predictor), whereas for larger data sets its impact can be neglected and the optimal performances exploiting only $\beta_3$ match those obtained optimizing jointly $\beta_2$ and $\beta_3$.
	\item The location of the minimum points (see also Fig. \ref{fig:Markovski_scatter}) is different depending on the employed training signal, especially in the low data regime; indeed, preference is given to exploiting the regularization term on $\gamma_2$ (and in fact $\gamma_3 = \infty$ in most cases).
\end{itemize}

In light of the above observations, the general validity of \eqref{eq:RHPC_prob_dd_gamma} constrained to either $\beta_2=0$ or $\beta_3=\infty$ devised in \cite{breschi2022uncertainty} is strengthened, as different types of data set are given. The effectiveness of these $\gamma$-DDPC schemes is also reinforced since it is evident that, in most cases, the operative tuning of either the sole parameter $\beta_2$ or the sole parameter $\beta_3$ is worth to be carried out in practice.

\begin{figure}[t!]
	\centering
  \subfigure[Training data set of size $10^4$]{\includegraphics[height=0.29\textwidth, trim={0.3cm 0cm 1cm 0cm},clip]{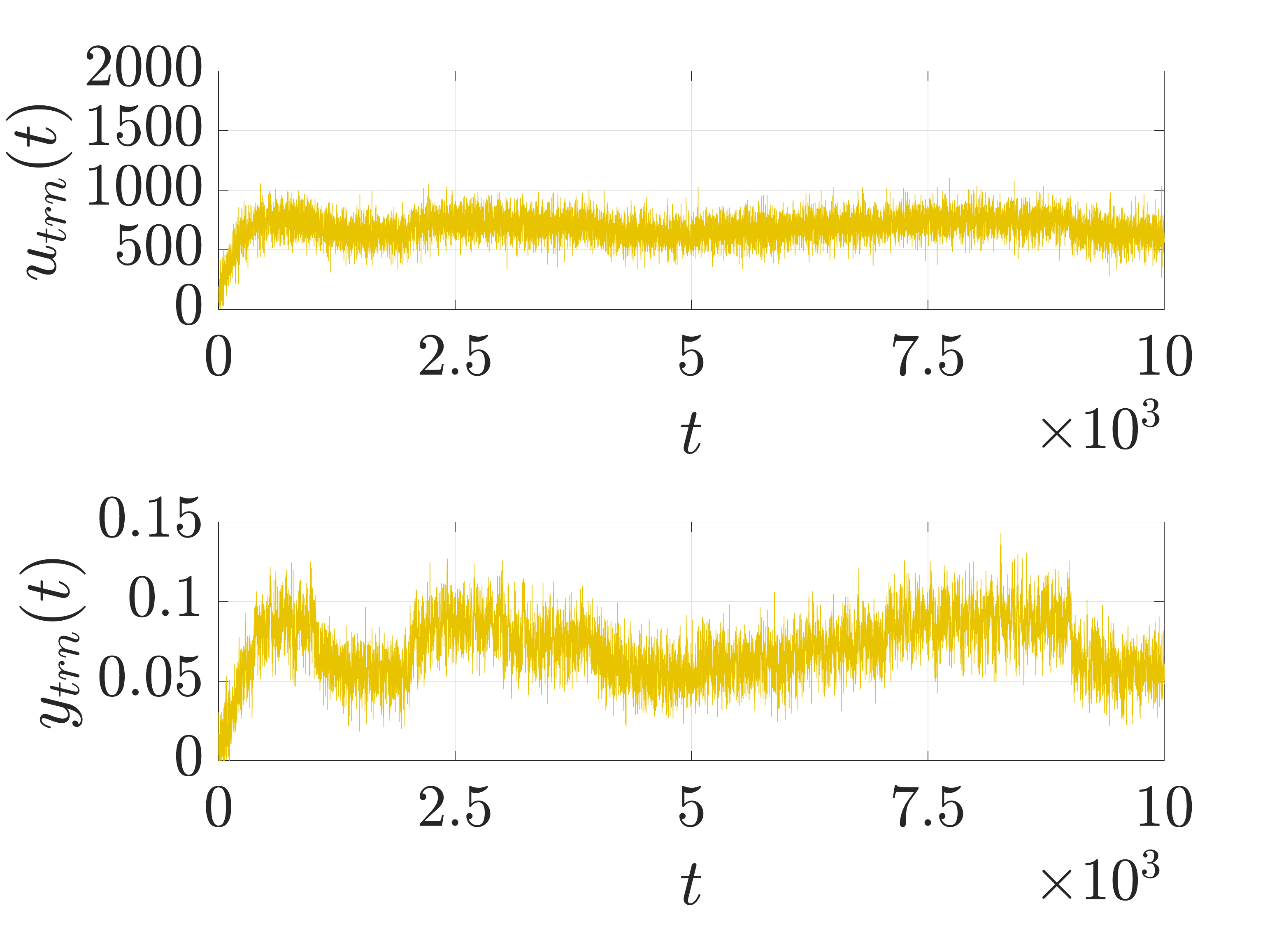}\label{fig:training}}
	\subfigure[Performance indexes]{\includegraphics[width=0.35\textwidth, trim={1.8cm 0cm 1cm 0.3cm},clip]{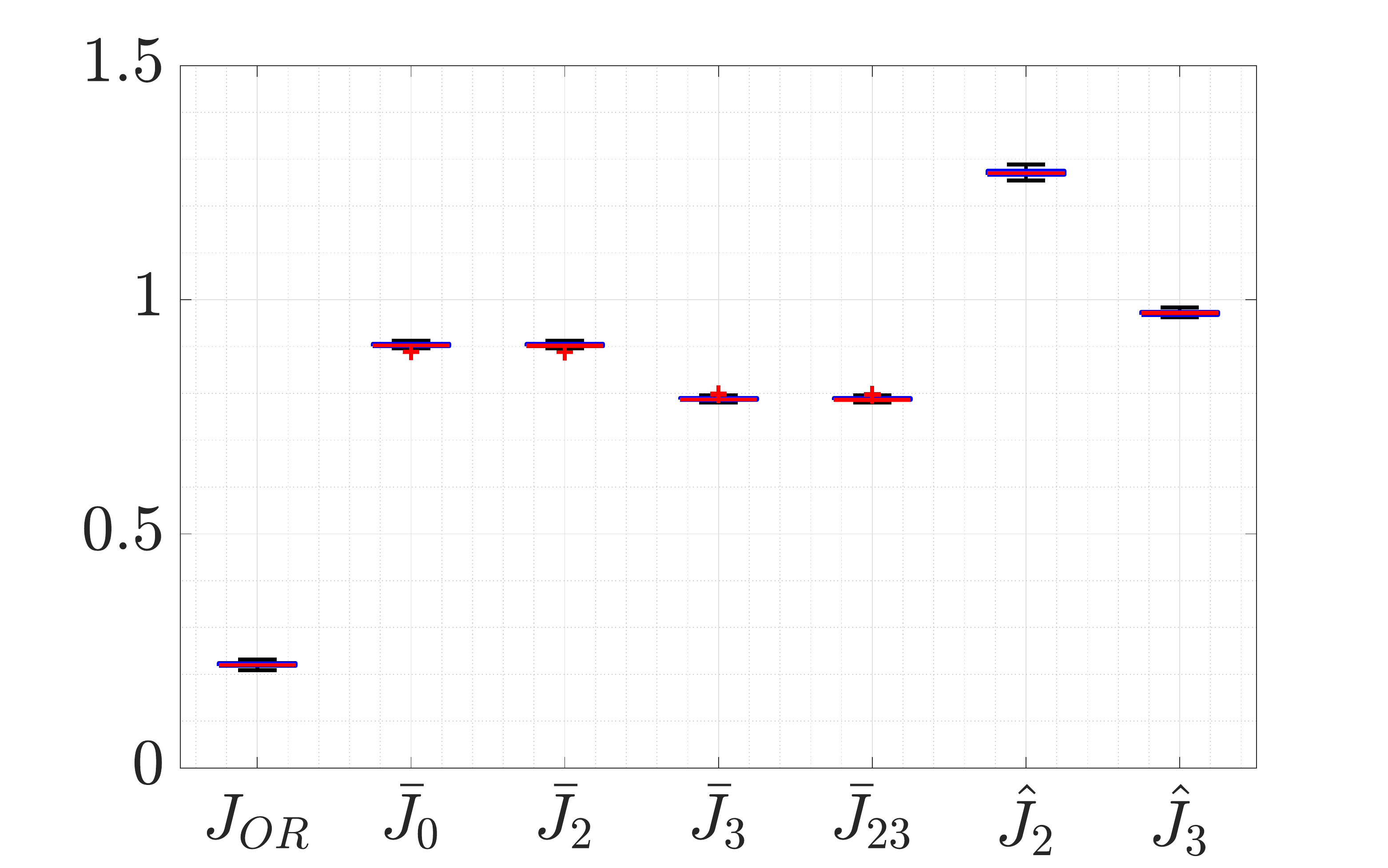}\label{fig:BS_perf_idx_exact}}
	\hspace{5mm}
	\subfigure[MPC-based oracle]{\includegraphics[height=0.26\textwidth, trim={0.3cm 0cm 1cm 0cm},clip]{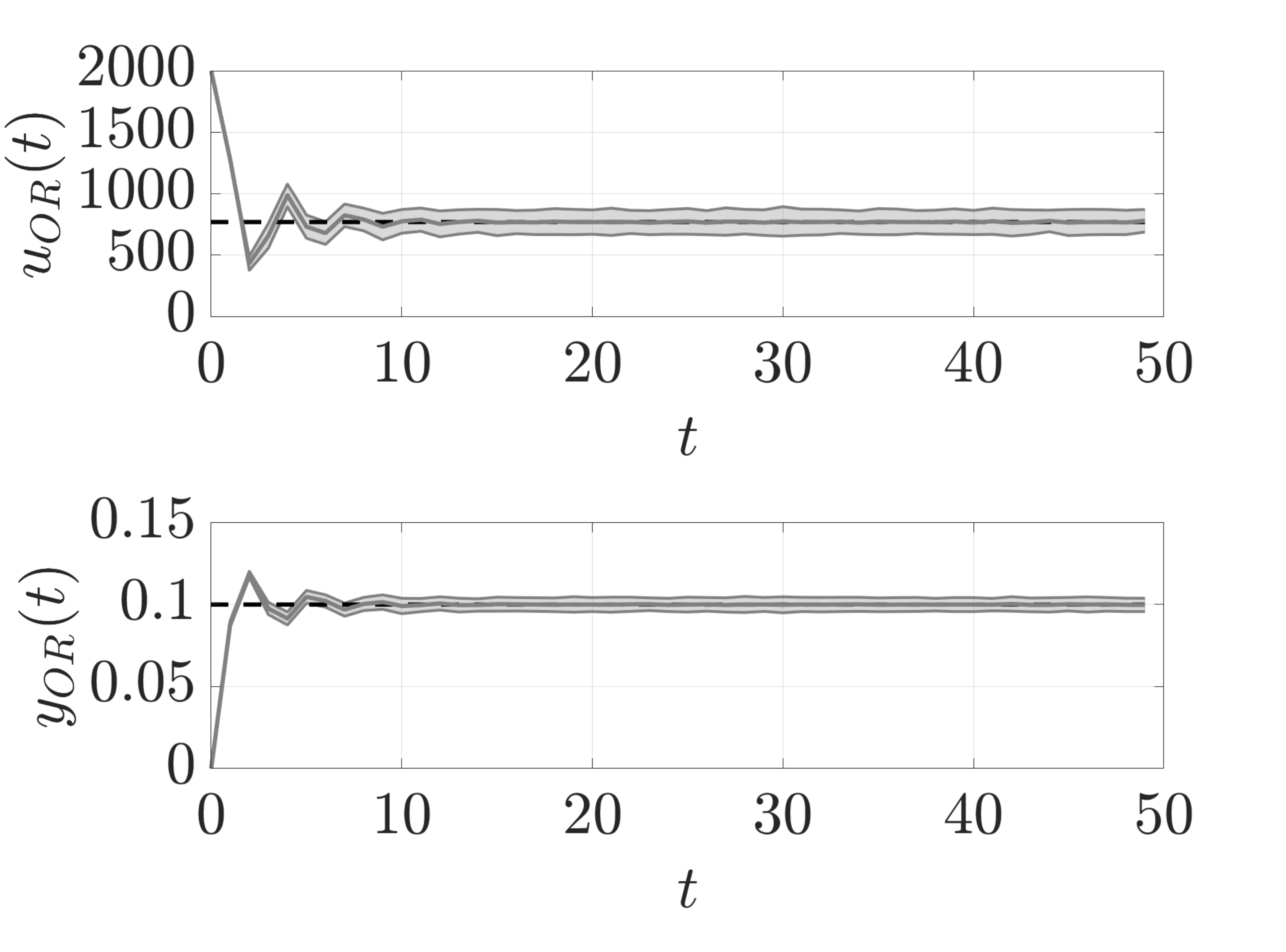}\label{fig:BS_oracle_inout_exact}}

	\caption{ (a): training data set employed for all the $\gamma$-DDPC experiments on the wheel slip control problem. 
	\rev{(b):} comparison of the performance indexes obtained with different $\gamma$-DDPC strategies (bar and hat notation indicating offline and online approaches respectively) and a model-based oracle. \rev{The subscript $a \in \{0, 2,3,23\}$ on $J$  refers to the regularization scheme (respectively: no regularization,  $\beta_2$, $\beta_3$ or both)}; 
	\rev{(c):} input/output tracking obtained from an MPC-based oracle. Mean (line) and $1.95$ times the standard deviation (shaded area) of the closed-loop input/output trajectories; the reference input and output are indicated with black dashed lines. 
 }
	\label{fig:BS_sims_tuned}
\end{figure}

\subsubsection{Wheel slip control problem}

\begin{figure*}[t!]
	\centering
	\subfigure[$\bar J_0$: $(\beta_2,\beta_3) = (0,+\infty)$]{\includegraphics[height=0.23 \textwidth, trim={0.3cm 0cm 1cm 0.2cm},clip]{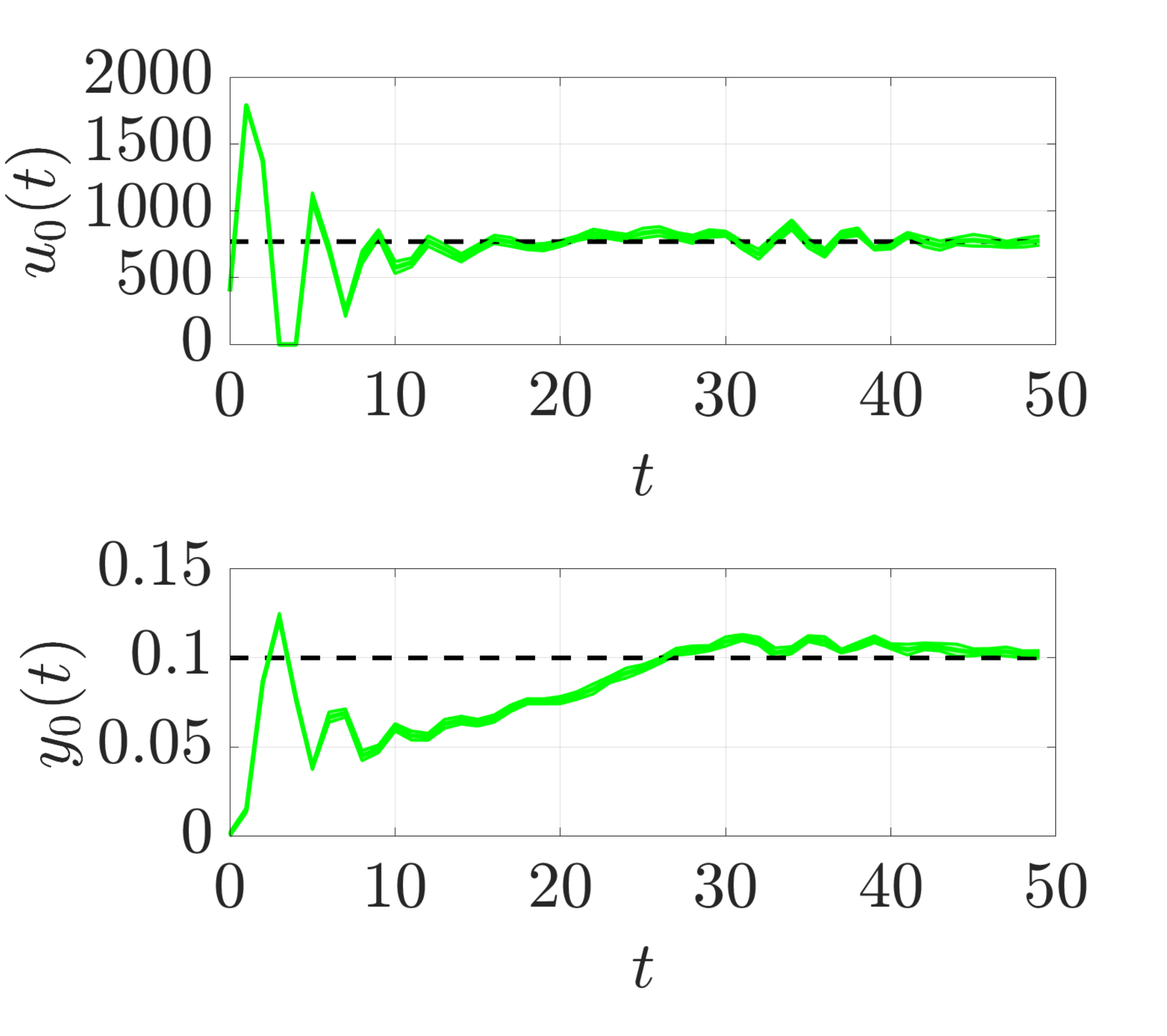}\label{fig:BS_noreg_inout_exact}}
	\hspace{5mm}
	\subfigure[$\bar J_2$: $(\beta_2,\beta_3)=(\bar{\beta}_2,+\infty)$]{\includegraphics[height=0.23\textwidth, trim={0.3cm 0cm 1cm 0cm},clip]{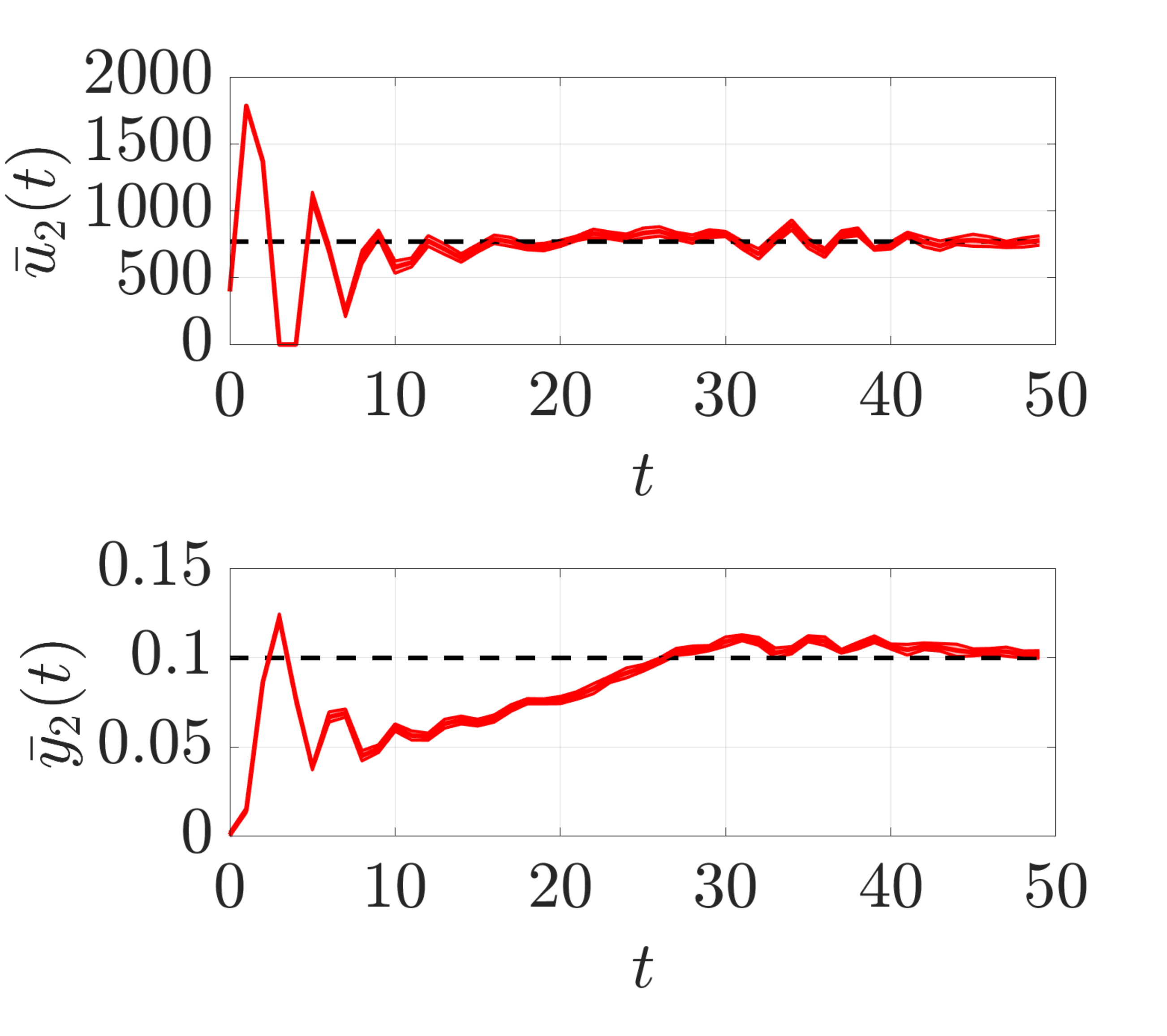}\label{fig:BS_b2bar_inout_exact}}
	\hspace{5mm}
	\subfigure[$\bar J_3$: $(\beta_2,\beta_3) = (0,\bar{\beta}_3$)]{\includegraphics[height=0.23\textwidth, trim={0.3cm 0cm 1cm 0cm},clip]{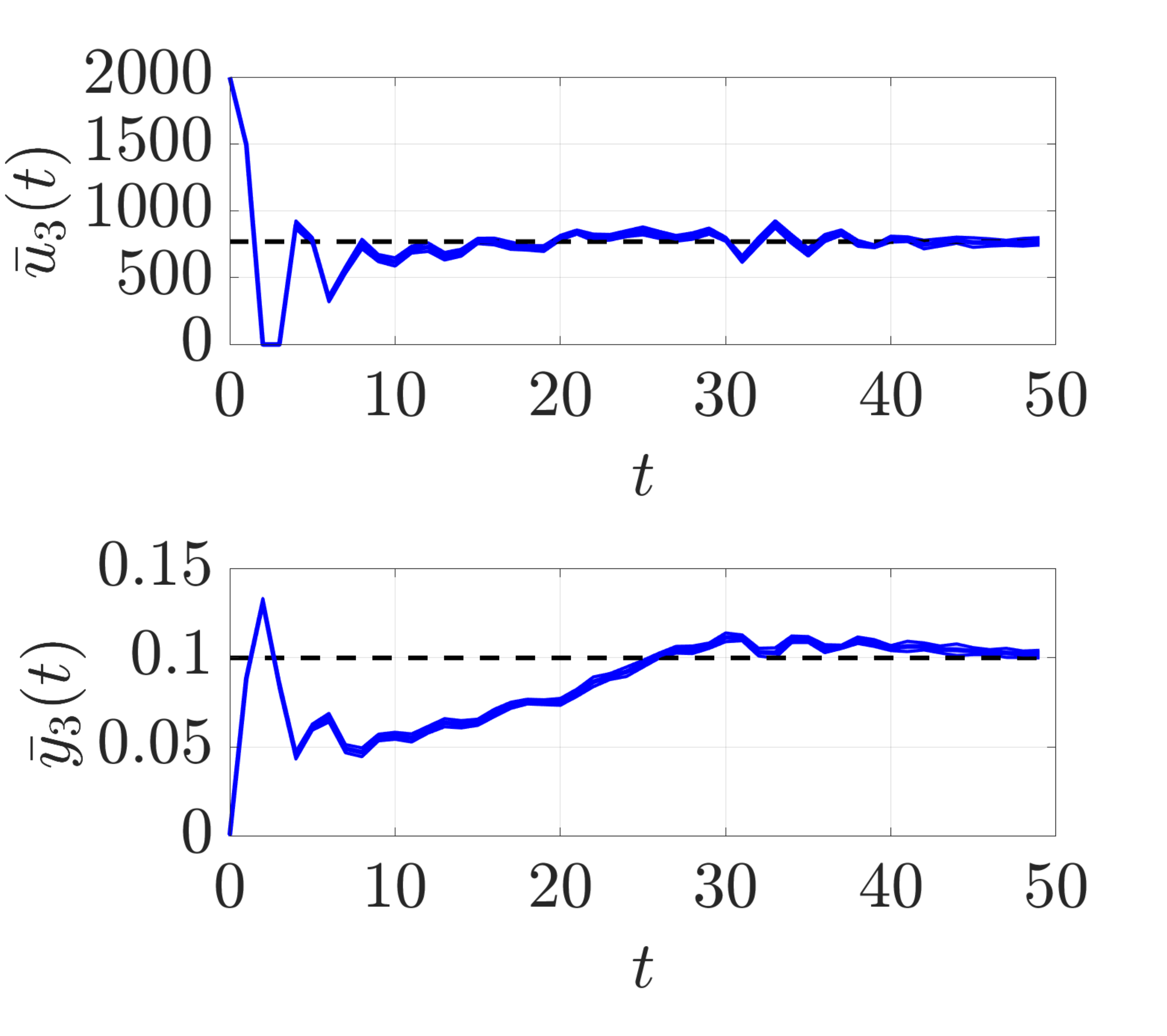}\label{fig:BS_b3bar_inout_exact}}
	\\
	\subfigure[$\bar J_{23}$: $(\beta_2,\beta_3) = (\beta^{\star}_2,\beta^{\star}_3)$]{\includegraphics[height=0.23\textwidth, trim={0.3cm 0cm 1cm 0.2cm},clip]{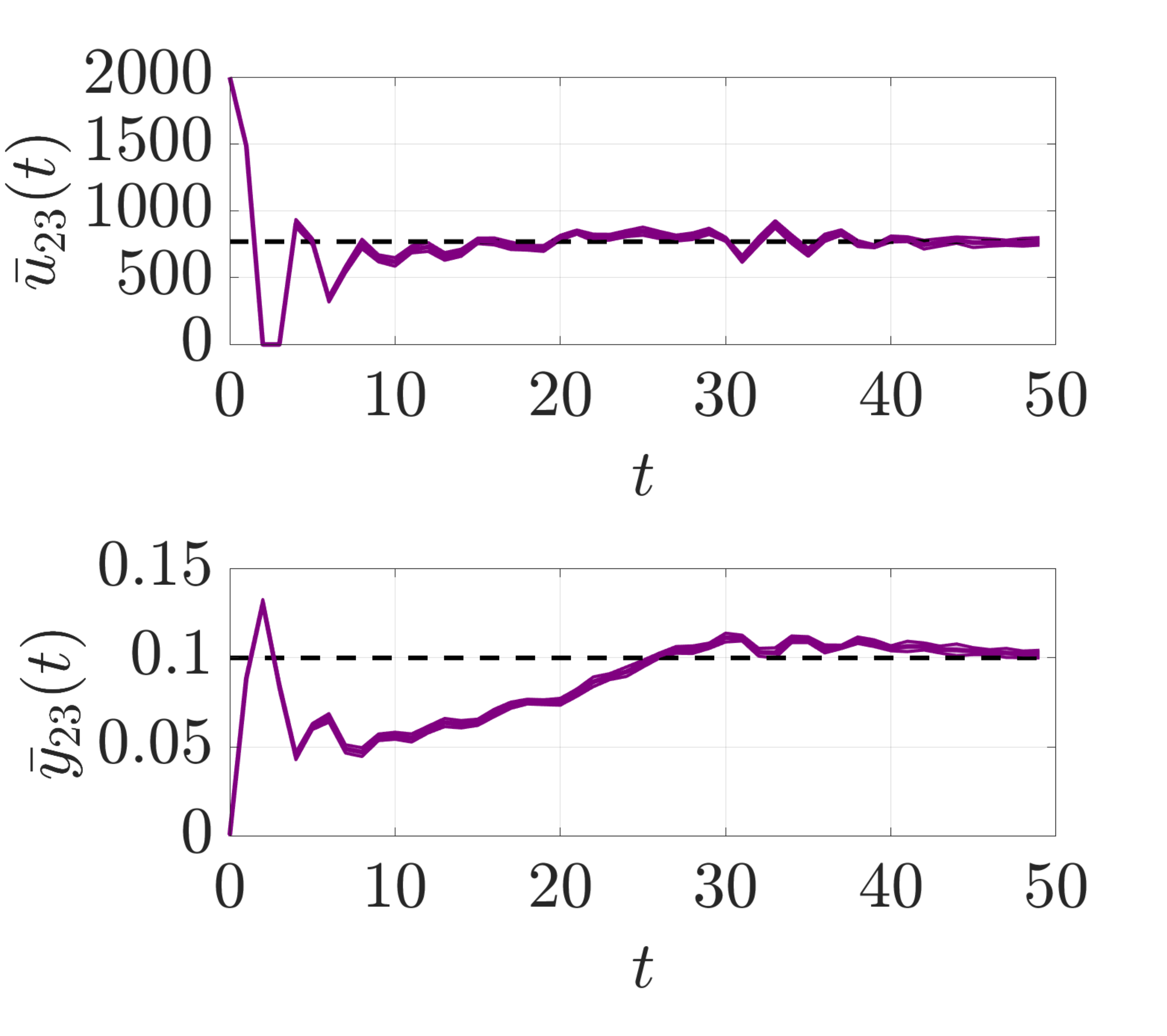}\label{fig:BS_b23bar_inout_exact}}
	\hspace{5mm}
	\subfigure[$\hat J_2$: $(\beta_2,\beta_3) = (\hat \beta_{2},+\infty)$]{\includegraphics[height=0.23\textwidth, trim={0.3cm 0cm 1cm 0cm},clip]{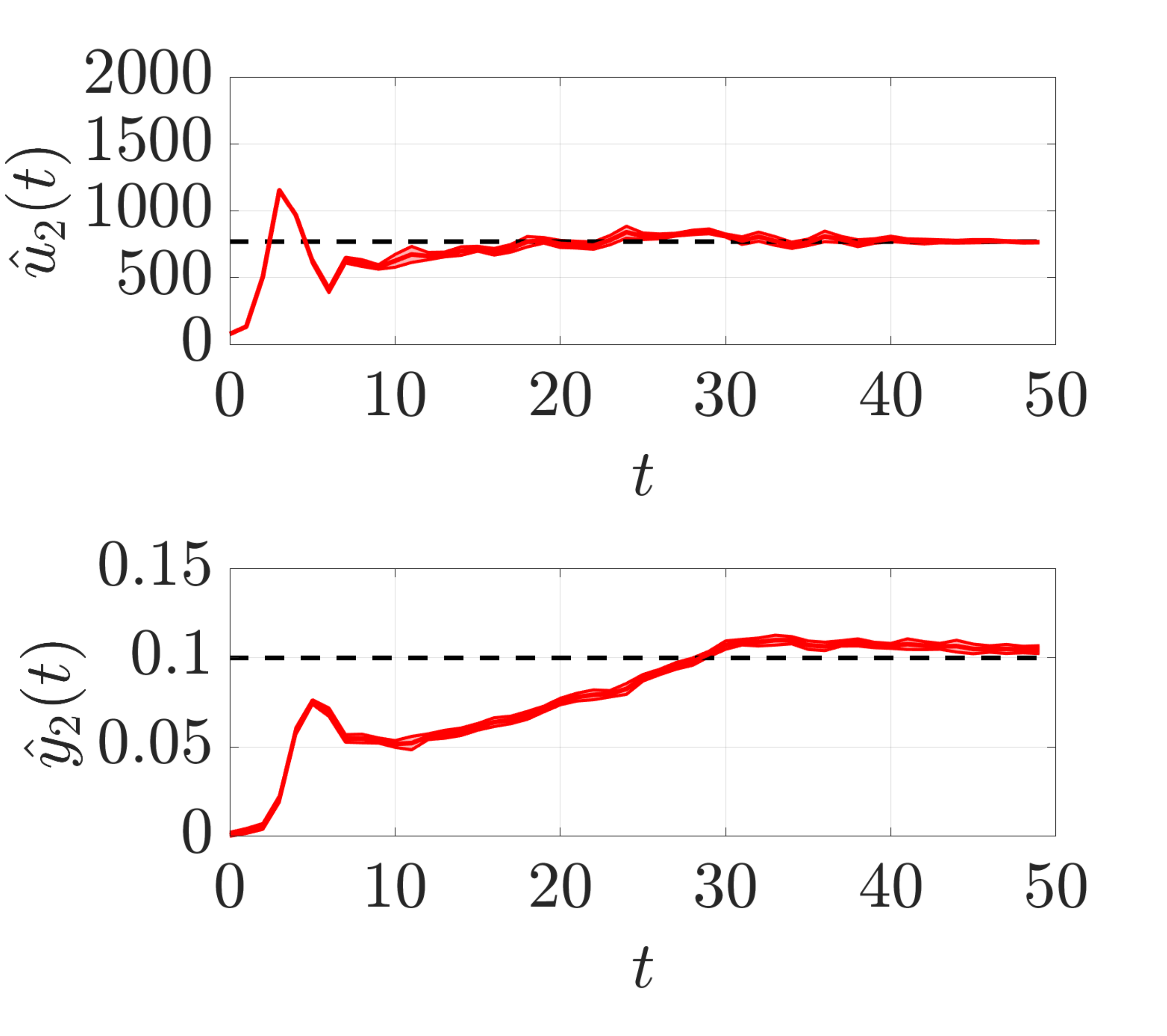}\label{fig:BS_b2tuned_inout_exact}}
	\hspace{6mm}
	\subfigure[$\hat J_3$: $(\beta_2,\beta_3) = (0,\hat \beta_{3})$]{\includegraphics[height=0.23\textwidth, trim={0.3cm 0cm 1cm 0cm},clip]{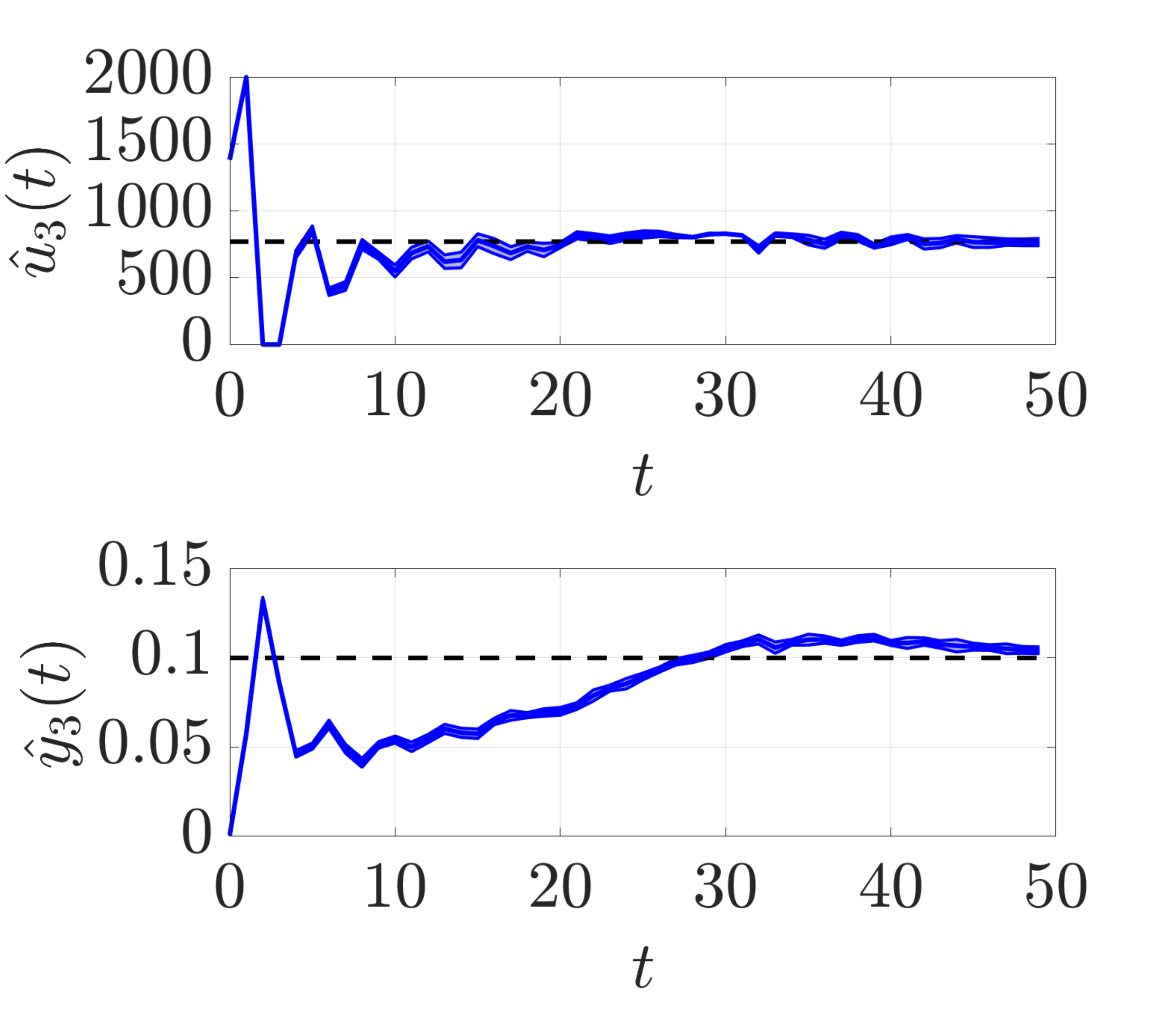}\label{fig:BS_b3tuned_inout_exact}}
	\caption{For all diagrams: mean (line) and $1.95$ times the standard deviation (shaded area) of the closed-loop input/output trajectories; the reference input and output are indicated with black dashed lines.
		(a): $\gamma$-DDPC with no regularization;
		(b)-(c): offline regularization strategies employing $\bar{\beta}_2$ and $\bar{\beta}_3$ separately;
		(d): offline regularization strategies employing $\bar{\beta}_2$ and $\bar{\beta}_3$ jointly;
		(e)-(f): online regularization strategies employing $\hat{\beta}_2$ and $\hat{\beta}_3$ separately.}
	\label{fig:BS_sims_new}
\end{figure*}

We now consider the problem of designing a wheel slip controller, steering the vehicle slip $\lambda(t) \in [0,1]$ to a constant target value $\lambda_{r}$. The design is carried out by focusing on quasi-stationary operating condition (the parameters of the vehicle, its velocity and the road profile are assumed to be constant). In both data collection and closed-loop testing, the behavior of the braking system (from now on indicated as $\Sigma_{NL}$) is simulated based on the 
\emph{nonlinear} 
\rev{model in} \cite{Formentin2015}:
\begin{subequations}\label{eq:slip_dynamics}
	\begin{equation}
		\dot{\lambda}(t)=-\frac{1}{v}\left(\frac{1-\lambda(t)}{m}+\frac{r^2}{J}\right)mg\mu(\lambda(t))+\frac{r}{Jv}T_{b}(t),
	\end{equation}
	where the road friction coefficient $\mu(\lambda(t))$ is assumed to be dictated by the Burckhardt model, \emph{i.e.,} 
	\begin{equation}\label{eq:Burckhardt_model}
		\mu(\lambda(t))=\alpha_{1}\left(1-e^{-\alpha_2\lambda(t)}\right)-\alpha_3 \lambda(t).
	\end{equation}
\end{subequations}
\rev{We indicate with} $T_b(t)$ [Nm] 
the \emph{controllable} braking torque and 
set the \rev{system} parameters 
to the same values used in \cite{sassella2022}. Although this dynamics is clearly nonlinear, it is possible to identify two main operating regions of the system\footnote{In this case, these two regions are limited by the slip $\bar{\lambda}=0.17$, for slip values lower than $0.17$; while it becomes unstable for higher slips.}, where the behavior of the slip can be approximated as linear. To comply with our framework (see \eqref{eq:stoc_sys}), we thus consider both data collection and simulation tests where the vehicle generally operates in a low-slip regime. Accordingly, data are gathered by performing closed-loop experiments with the benchmark controller introduced in \cite{savaresi2010active}, selecting a slip reference uniformly chosen at random in the interval $[0,0.15]$ and collecting data at a sampling rate of $100$~[Hz].
In particular, the output $y_{trn}(t)$ of the employed training data set shown in Fig. \ref{fig:training} 
is generated by exploiting a closed-loop experiment wherein the output is corrupted by a zero-mean white noise process with variance $\sigma^2_n = 10^{-6}$ and, also, zero-mean white noise with variance $10^8 \sigma^2_n$ is added to the input $u_{trn}(t)$ provided by the controller.

Meanwhile, the reference slip for the closed-loop tests is $\lambda_{r}=0.1$, corresponding to a reference braking torque $T_{b,r}=768.9$ [Nm]. To improve the tracking performance in closed-loop, apart from the terms weighting the tracking error and the difference between the predicted and reference torque, respectively weighted by $Q = 10^{3}$ and $R = 10^{-7}$, the cost of the $\gamma$-DDPC problem \eqref{eq:cost_gammaDDPC} is augmented with a term penalizing abrupt variations of the input (weighted as $10^{-4}$), a term penalizing the integral of the tracking error (weighted as $10^5$), and two terms further penalizing the difference between the slip and torque references and their actual value over the last step of the prediction horizon (weighted as $10^3$ and $5 \cdot 10^{-6}$, respectively). The following constraint is also added at each feedback step $t\geq 0$ for $s = 0,\ldots,T-2$:
\begin{equation}
	\begin{cases}
		q(t+s) = q(t+s-1) + y_r(t+s-1)-\hat{y} (t+s-1);\\
		q(t-1) = y_r(t-1)-y(t-1);
	\end{cases}
\end{equation}
to account for the known dynamics of the integrator. Nonetheless, performances are still assessed via the index in \eqref{eq:general_perf_index} over a closed-loop test of $T_{v}=50$ steps. 
A Monte Carlo campaign with $100$ iterations is run on the above setup, corrupting the 
output of $\Sigma_{NL}$ 
with a white noise having signal-to-noise ratio $40$ dB. For each of the $100$ tests, the regularization parameters $\beta_2$ and $\beta_3$ are both selected from a grid over $[10^{-4},10^4]$ comprising of $15$ logarithmic-spaced points. For the joint optimization, the squared grid $\{\bar{\beta}_2,\bar{\beta}_2/10,\bar{\beta}_2 /100, 10^8\} \times \{ 0,\bar{\beta}_3,10\bar{\beta}_3,100\bar{\beta}_3\}$ composed by the optimal values $(\bar{\beta}_2,\bar{\beta}_3)$ obtained via offline $\gamma$-DDPC is instead taken into account.
Fig. \ref{fig:BS_perf_idx_exact} depicts the distributions of the 
performance index in \eqref{eq:general_perf_index} as the selected regularization strategy varies considering $(\beta_2,\beta_3)$ tuned either offline or online and comparing $\gamma$-DDPC with a MPC-based oracle (see also Fig. \ref{fig:BS_oracle_inout_exact}). In particular, the input-output trajectories of all $\gamma$-DDPC strategies can be summarized in Fig. \ref{fig:BS_sims_new}. Although the MPC-based oracle displays evident preeminence, it is worthwhile to appreciate that all these trajectories are characterized by solid performances (rise time of at most $5$ steps, settling time of about $25$ steps, maximum overshoot of $40\%$ or less with no cross into the unstable region), with such traits indicating that $\gamma$-DDPC schemes remain competitive even in nonlinear scenarios. Noticeably, the online strategies (implementable on real applications) shown in Fig. \ref{fig:BS_b2tuned_inout_exact}-\ref{fig:BS_b3tuned_inout_exact} share similar performances with the corresponding offline strategies in Fig. \ref{fig:BS_b2bar_inout_exact} - \ref{fig:BS_b3bar_inout_exact}, especially that relying on the online tuning of parameter $\beta_3$. Moreover, within the setup of this numerical example, one observes that the performance of the offline strategy based on $\beta_2$ strictly matches with that of the scheme lacking of regularization (Fig. \ref{fig:BS_noreg_inout_exact}); whereas, the performance of the offline strategy based on $\beta_3$ strictly matches that of the scheme in which a joint optimization of both $\beta_2$ and $\beta_3$ (Fig. \ref{fig:BS_b23bar_inout_exact}) is carried out. Hence, under this setup and with the data collected in this numerical example, it emerges once again that the optimal tuning based on the sole penalty parameter $\beta_3$ (i.e., setting $\beta_2 = 0$) can be considered \textit{in practice} for high-data regimes. This, in turn, may lead to significantly diminish the computational burden associated to the tuning of the penalty parameters whenever a real implementation based on the proposed regularized scheme \eqref{eq:RHPC_prob_dd_gamma} is considered and a big training data set is available.

The above comparison further highlights that regularized DDPC approaches can be competitive w.r.t. traditional model-based controllers and that $\gamma$-DDPC solution with the online tuning proposed in \cite{breschi2022uncertainty} can be robustly effective also when dealing with nonlinear systems.


\section{Concluding remarks and future directions}\label{sec:conclusions}
Several  regularization strategies for Data Driven Predictive Control ($\gamma$-DDPC) have been discussed and evaluated in terms of  closed-loop performance. It has been proved that when the input is white, regularizing $\gamma_2$ and penalizing control energy are equivalent.  Numerical examples further illustrate that the tuning of the penalty parameters in the $\gamma$-DDPC can be decoupled without dramatically impacting the performance corresponding to a (more costly) joint regularization wherein both $\beta_2$ and $\beta_3$ are accounted for. 

Future work will be devoted to the extensive experimental assessment of the considered regularization strategies, as well as to a 
theoretical analysis of the optimization of the sole $\beta_3$.

\bibliographystyle{IEEEtran}
\bibliography{biblio}

\end{document}

%% file: shortcuts_math_CDC.tex
\def \u {\mathbf{u}}

\def \U {\mathbf{U}}